\documentclass[copyright,creativecommons]{eptcs}

\providecommand{\event}{VPT 2021} 
\usepackage{underscore}  

\usepackage[T1]{fontenc}
\usepackage[scaled]{}
\usepackage[english]{babel}
\usepackage[utf8]{inputenc}
\usepackage{listings}
\usepackage{mathrsfs, amsfonts, amsmath, amssymb, latexsym}
\usepackage{syntax}
\usepackage{multicol}
\usepackage{color}
\usepackage{graphicx}
\usepackage{graphics}
\usepackage{listings}
\usepackage{stmaryrd}
\usepackage{galois}
\usepackage[ruled, linesnumbered]{algorithm2e}
\usepackage{capt-of}
\usepackage{makecell}
\usepackage{framed}
\usepackage{float}
\usepackage{mathbbol}

\usepackage{graphicx}
\usepackage{subcaption}
\usepackage{caption}

\usepackage{pgf}
\usepackage{tikz, wrapfig,fancybox}
\usetikzlibrary{arrows,automata}
\usepackage[export]{adjustbox}

\makeatletter
\newenvironment{CenteredBox}{%
	\begin{Sbox}}{
	\end{Sbox}\centerline{\parbox{\wd\@Sbox}{\TheSbox}}}
\makeatother

\makeatletter
\@ifpackageloaded{stix}{%
}{%
	\DeclareFontEncoding{LS2}{}{\noaccents@}
	\DeclareFontSubstitution{LS2}{stix}{m}{n}
	\DeclareSymbolFont{stix@largesymbols}{LS2}{stixex}{m}{n}
	\SetSymbolFont{stix@largesymbols}{bold}{LS2}{stixex}{b}{n}
	\DeclareMathDelimiter{\lBrace}{\mathopen} {stix@largesymbols}{"E8}%
	{stix@largesymbols}{"0E}
	\DeclareMathDelimiter{\rBrace}{\mathclose}{stix@largesymbols}{"E9}%
	{stix@largesymbols}{"0F}
}
\makeatother

\SetKwInput{KwData}{Input}
\SetKwInput{KwResult}{Output}

\lstdefinelanguage{JS}
{ keywords={var, function, typeof, this, undefined, parseInt},
	otherkeywords={},
	basicstyle=\fontsize{10pt}{12pt}\selectfont\ttfamily,
	keywordstyle=\bfseries\color{blue},
	sensitive=false,
	commentstyle=\color{purple!40!black},
	showspaces=false,
	tabsize=1,
	literate= {~}{\texttt{\phantom{m}}}1 {`}{$\texttt{\%}$}1 {?}{$\texttt{\$}$}1, 
	showstringspaces=false,emph={3}{\tiny }
	showtabs=true,
	morecomment=[l]{//},
	morecomment=[s]{/*}{*/},
	morestring=[b],
	breaklines=true,
	breakindent=12pt
}

\definecolor{dkred}{rgb}{.6,0,0}
\definecolor{dkgreen}{rgb}{0,.5,0}
\definecolor{dkblue}{rgb}{0,0,.6}
\definecolor{dkyellow}{cmyk}{0,0,.8,.3}
\definecolor{lightgray}{rgb}{.95,.95,.95}
\definecolor{darkgray}{rgb}{.3,.3,.3}
\definecolor{darkblue}{rgb}{0,0,.20}
\definecolor{purple}{rgb}{0.65, 0.12, 0.82}

\newif\ifdraft\drafttrue

\newcommand{\green}[2]{{\color{dkgreen}(#1: #2)}}
\newcommand{\xremark}[2]{{\color{dkred}(#1: #2)}}
\newcommand{\blu}[2]{{\color{purple}(#1: #2)}}
\ifdraft
\newcommand{\va}[1]{\xremark{VA}{#1}}
\newcommand{\im}[1]{\green{IM}{#1}}
\newcommand{\del}[1]{\blu{VA}{#1}}
\else
\newcommand{\va}[1]{}
\newcommand{\im}[1]{}
\newcommand{\del}[1]{}
\fi

\newcommand{\red}[1]{{\color{dkred}#1}}

\newcommand{\sset}[2]{\left\{~#1  \left |
	\begin{array}{l}#2\end{array} \right.\!\right\}}

\newcommand{\CAbs}{\Upsilon}

\newcommand{\cstr}{\mathbb{S}}

\newcommand{\cint}{\mathbb{Z}}
\newcommand{\imp}{\CommS}

\renewcommand{\exp}{\mbox{\sf e}}

\newcommand{\bexp}{\mbox{\sf b}}

\newcommand{\prog}{{\mbox{\sf P}}}

\newcommand{\concat}[2]{\mbox{\tt concat($#1$,$#2$)}}

\newcommand{\true}{{\tt true}}
\newcommand{\false}{{\tt false}}

\newcommand{\pp}[1]{{{}^{\red{\ell_{#1}}}}}

\newcommand{\rhoimp}{\rho_{\scriptscriptstyle \mathcal{CS}}}
\newcommand{\tuple}[1]{\langle #1\rangle}

\newcommand{\lab}{\mathsf{lab}}

\def\rarr#1{\mbox{\raisebox{0ex}[1ex][1ex]{$
			\mathrel{\mathop{
					\hspace*{1pt}\longrightarrow\hspace*{1pt}}\limits^{\,_{\mbox{\tiny 
							\hspace*{-2.2pt}#1}}}}$}}}
\def\ok#1{\mbox{\raisebox{0ex}[1ex][1ex]{$#1$}}} 
\newcommand{\ra}{\rightarrow}

\newcommand{\Lra}{\Leftrightarrow}
\newcommand{\lra}{\longrightarrow}
\newcommand{\uco}{\mbox{\it uco}}
\def\defi{\mbox{\raisebox{0ex}[1ex][1ex]{$\stackrel{\mbox{\tiny
					def}}{\; =\;}$}}}

\def\Sign{\mathsf{Sign}}

\def\grasse#1#2{[\![#1]\!]^{#2}}

\def\defemb#1#2{\expandafter\def\csname #1\endcsname
	{\relax\ifmmode #2\else\hbox{$#2$}\fi}}

\def\cS{\mathcal{S}}
\def\cL{\mathcal{L}}

\def\cP{\mathcal{P}}

\newcommand{\rhos}{\rho_{\mbox{\tiny\sl S}}}
\newcommand{\sval}[2]{\cS^{#1}(#2)}

\newcommand{\alphabet}{\mathcal{K}}

\newcommand{\cval}{\mathbb{V}}

\newcommand{\Exp}{{\sf E}}

\newcommand{\rMem}[1]{\Mem^{#1}}

\newcommand{\code}[1]{\mbox{\lstinline@#1@}}

\newcommand{\interp}[1]{\grasse{#1}}

\newcommand{\Lab}[1]{\mbox{\sl Lab$_{#1}$}}
\newcommand{\ccc}[1]{}
\newcommand{\CommS}{\textsf{Imp}}
\def\tG{\mathtt{G}}

\def\aetG{\tG^\eta}
\newcommand{\edges}{\mbox{\sl Edges}}
\newcommand{\aedges}{\mbox{\sl Edges}^\#}
\newcommand{\nodes}{\mbox{\sl Nodes}}

\def\CFG{\mathsf{CFG}}
\def\aCFG{\mathsf{CFG}^\#}

\newcommand{\ppl}[1]{{{}^{\red{\ell_{#1}}}}}
\def\skipc{\mbox{\bf skip}}
\def\Exp{\mbox{\tt e}}
\def\Aexp{\mbox{\tt a}}
\def\Bexp{\mbox{\tt b}}
\def\Sexp{\mbox{\tt s}}
\def\Comm{\mbox{\tt c}}
\newcommand{\CommG}{\Psi}
\def\tl{\mathtt{l}}
\def\Mem{\mathbb{M}}

\newcommand{\mem}{\mathbb{m}}

\newcommand{\rmem}[1]{\mem^{#1}}
\newcommand{\Var}{\ensuremath{\textsf{Var}}}
\newcommand{\Int}{\ensuremath{\mathbb{Z}}}
\newcommand{\Bool}{\ensuremath{\mathsf{Bool}}}
\def\grasseb#1{{\llparenthesis\hspace*{0.2ex} #1 \hspace*{0.2ex}\rrparenthesis}}
\def\grasse#1{{\llbracket #1 \rrbracket}}
\def\P{\tt P}
\newcommand{\AexpS}{\mbox{AExp}}
\newcommand{\BexpS}{\mbox{BExp}}
\newcommand{\ExpS}{\mbox{Exp}}

\newcommand{\integer}{\mathbb{Z}}

\def\sep{;}

\def\ov{\overline}

\def\Sexp{{\sf s}}
\DeclareMathOperator{\SexpS}{SExp}
\newcommand{\mstr}[1]{\mbox{\tt "$#1$"}}
\newcommand{\subst}[3]{\code{substr}(#1,#2,#3)}
\def\reflect{\code{eval}}
\DeclareMathOperator{\Ccomms}{Comm}
\def\tc{\mathtt{c}}

\def\tG{\mathtt{G}}

\newcommand{\heta}{\widehat{\rho}}

\newtheorem{theorem}{Theorem}[section]
\newtheorem{proposition}[theorem]{Proposition}
\newtheorem{example}[theorem]{Example}
\newtheorem{corollary}[theorem]{Corollary}

\newtheorem{lemma}[theorem]{Lemma}

\newenvironment{proof}{\noindent {\sc Proof.~}}{\hfill $\Box$\newline\smallskip\mbox{}\unskip~~\normalsize}

\begin{document}
\title{Improving Dynamic Code Analysis by Code Abstraction}
\author{Isabella Mastroeni
\institute{Department of Computer Science, University of Verona (Italy)}
\email{isabella.mastroeni@univr.it}
\and
Vincenzo Arceri
\institute{Department of Environmental Sciences, Informatics and Statistics,\\ Ca' Foscari University of Venice (Italy)}
\email{vincenzo.arceri@unive.it}
}
\def\titlerunning{Improving Dynamic Code Analysis by Code Abstraction}
\def\authorrunning{I. Mastroeni \& V. Arceri}
%

\maketitle              
\begin{abstract}
In this paper, our aim is to propose a model for code abstraction, based on abstract interpretation, allowing us to improve the precision of a recently proposed static analysis by abstract interpretation of dynamic languages. The problem we tackle here is that the analysis may add some spurious code to the string-to-execute abstract value and this code may need some abstract representations in order to make it analyzable. This is precisely what we propose here, where we drive the code abstraction by the analysis we have to perform.

\end{abstract}

\section{Introduction}
The possibility of dynamically building code instructions as the result of text manipulation is a key aspect in dynamic programming languages. In this scenario, programs can turn text, which can be built at run-time, into executable code \cite{RichardsHBV11}. These features are often used in code protection and tamper-resistant applications, employing camouflage for escaping attack or detection \cite{DMavrogiannopoulosKP11}, in malware, in mobile code, in web servers, in code compression, and in code optimization, e.g., in Just-in-Time (JIT) compilers, employing optimized run-time code generation.\\
While the use of dynamic code generation may simplify considerably the {\em art and performance of programming\/}, this practice is also highly dangerous, making the code prone to unexpected behaviors and malicious exploits of its dynamic vulnerabilities, such as code/object-injection attacks for privilege escalation, database corruption, and malware propagation. It is clear that more advanced and secure functionalities based on string-to-code statements could be permitted if we better master how to safely generate, analyze, debug, and deploy programs that dynamically generate and manipulate code.

There are lots of good reasons to analyze programs building strings that can be later executed as code. 
An interesting example is code obfuscation. Recently, several techniques have been proposed for JavaScript code obfuscation\footnote{\url{https://www.daftlogic.com/projects-online-javascript-obfuscator.htm},\\\url{http://www.danstools.com/javascript-obfuscate/},\\\url{http://javascript2img.com/},\\\url{https://javascriptobfuscator.herokuapp.com/},\\\url{https://javascriptobfuscator.com/}}, 
meaning that also client-side code protection is becoming an increasingly important problem to be tackled by the research community and by practitioners. Hence, it is not always possible to simply ignore \code{eval} without accepting to lose the possibility of analyzing the rest of the program \cite{tops20}.
\paragraph*{The Context: Analyzing Dynamic Code.}
A major problem in presence of dynamic code generation is that static analysis becomes extremely hard if not  impossible. This happens because program data structures, such as the control-flow graph and the system of recursive equations associated with the program in question, are themselves dynamically mutating objects. 
Recently \cite{tops20}, the problem of analyzing dynamic code has been tackled by {\em treating code as any other dynamic structure that can be statically analyzed by abstract interpretation, and to treat the abstract interpreter as any other program function that can be recursively called}.
In particular, in \cite{tops20}, we provide a static analyzer architecture for a core dynamic language, containing non-removable \code{eval} statements, that still has some limitation in terms of precision but  
provides the necessary ground for studying more precise solutions to the problem. In particular,
\begin{itemize}
\item[$\bullet$] We have designed an automata-based string abstract domain \cite{mdpi} for analyzing string values during execution. Automata (FA) provide the perfect choice for abstracting strings that may be executed by \code{eval} since they allow us to over-approximate the set of possible values of string variables by keeping enough information for both analyzing properties of string variables that are never executed by an \code{eval} during computation and for extracting the potential executable sub-language. 
\item[$\bullet$] In order to statically analyze the code potentially executed by an \code{eval}, we have designed a systematic process for extracting from the (abstract) argument of \code{eval} (i.e., from the FA collection of its potential arguments) an over-approximation of executable code that this collection contains. Clearly, this approximation must keep a form that the analyzer can interpret.
\item[$\bullet$] We designed a static analyzer for dynamic languages performing a recursive call of the interpreter on the (over-approximated) code that \code{eval} may execute.
\end{itemize}
\paragraph*{The Problem: Improve Precision Analysis by Abstracting Code.}
This analysis provides a first step towards the analysis of dynamic languages but still has some important precision loss \cite{tops20}.
In particular, there are particular forms of FA (which occur when the string is dynamically generated by loops) avoiding the possibility of generating a control flow graph ($\CFG$) able to approximate the code executed by an \code{eval}.
%
%
For instance, when the FA accepts a language such as $\sset{\code{x=(5)}^n\code{;}}{n > 0}$, the analysis in \cite{tops20} cannot extract, from the FA, the $\CFG$ approximating the \code{eval} argument. 
In order to better explain the problem, consider the code in Fig.~\ref{fig:limit}, where the value of $\code{i}$ is statically unknown. 
In Fig.~\ref{fig:limit}, we draw the automaton $A$ representing the abstract value of {\tt str} before the $\code{eval}$ execution. The problem is that $A$ has a cycle not involving a whole statement~\cite{tops20}. This situation makes the analyzer unable to build a $\CFG$ over-approximating the code potentially executed since, intuitively, such a $\CFG$ should be infinite. Indeed, only an infinite $\CFG$ could capture all the possible assignments described by the FA, namely all the assignments of any possible number formed only by ${\tt 5}$ to the variable ${\tt x}$ (i.e., {\tt x=5;},{\tt x=55;},{\tt x=555;}$\dots$). \\
\begin{figure}[t]
	\begin{subfigure}[b]{0.5\textwidth}
		\begin{CenteredBox}
			\begin{lstlisting}[basicstyle=\fontsize{10}{10}\selectfont\ttfamily,escapeinside={(*}{*)}]
str = "x=5";
while (i < 3) {
	str = str + "5";
	i = i + 1;
}
str = str + ";"; eval(str);	
			\end{lstlisting}
		\end{CenteredBox}
	\end{subfigure}~
	\begin{subfigure}[b]{0.5\textwidth}
		\centering
		\begin{adjustbox}{scale=0.8} 
			\begin{tikzpicture}[->,>=stealth',shorten >=1pt,auto,node distance=1.7cm, semithick]
			\node[initial,state,scale=0.6, initial text =] 			 	 (A)                    {};
			\node[state,scale=0.6]      			 			 (B) [right of=A] 		{};
			\node[state,scale=0.6]        			 		(C) [right of=B] 		{}; 		
			\node[state,scale=0.6]        			 		(D) [right of=C] 		{}; 		
			\node[state,scale=0.6, accepting]        			 		(E) [right of=D] 		{}; 
			
			\path[->] (A) edge node {\code{x}} (B);
			\path[->] (B) edge node {\code{=}} (C); 	
			\path[->] (C) edge node {\code{5}} (D); 		
			\path[->] (D) edge  [loop above] node {\code{5}} (D); 	
			\path[->] (D) edge node {\code{;}} (E);
			\end{tikzpicture}
		\end{adjustbox}
	\end{subfigure}
	\caption{$A$ s.t. $\cL(A) = \{\mathtt{x=5}^n\mathtt{;}~|~n > 0\}$, where $\mathtt{5}^n$ means $\mathtt{5}$ repeated $n$ times.}
	\label{fig:limit}
\end{figure}
In order to make it possible to overcome this limitation, at least for a set of potential \code{eval} patterns, we propose to define a form of {\em abstract $\CFG$} able to finitely represent a potential infinite set of $\CFG$s, e.g., we look for  a $\CFG$ representing {\tt x=5$^*$}.\\ Unfortunately, things are not so easy as it may seem, since this abstract code representation has to be built in such a way that the analyzer may still be able to interpret it. 
\paragraph*{Contribution.} 
The main contributions for tackling the problem above are:
\begin{itemize}
\item[$\bullet$] We first define the notion of {\em abstract $\CFG$}, based on the idea of making it possible to still perform a given analysis. The idea is to leave the control structure unchanged while approximating the edge labels (the statements to execute) to sets of labels, i.e., those sharing a fixed abstract property.
\item[$\bullet$] We show how completeness of code abstraction w.r.t.\ the semantic observation models the possibility, for the static analyzer, of interpreting also the abstract code, and we show how we can make any code abstraction complete. 
\item[$\bullet$] We provide a systematic approach, based on the one proposed in \cite{tops20}, allowing us to analyze also the \code{eval} patterns described above, for which, instead, the analysis in \cite{tops20} loses precision.
\end{itemize}

\section{The Core Language: $\imp$}\label{sect:bg}
The language is quite standard (see Fig.~\ref{synta}\footnote{We use $n$ to denote the semantic value corresponding to the syntactic symbol $\code{n}$.}), and each statement is annotated with a label $\ell \in \Lab{}$ (not part of the syntax) corresponding to the statement program point\footnote{
We suppose that there exists a function that, taken a well-written program, can label it with a fresh label for each program point.}.
\begin{figure}[t]
\hspace{-2cm}
{\small  
  \begin{align*}
  \ExpS\ni \Exp &::= \; \Aexp\mid \Sexp\\
   \AexpS \ni \Aexp  &::= \; \code{x} \mid \code{n} \mid 
    \Aexp + \Aexp \mid \Aexp -\Aexp \mid \Aexp*\Aexp\\
    \BexpS \ni \Bexp &::= \; \code{x} \mid\true\mid\false \mid \Exp=\Exp \mid 
    \Exp > \Exp \mid \Exp < \Exp \mid \Bexp\wedge\Bexp \mid \neg {\tt b} \\
    \SexpS\ni \Sexp &::= \; \code{x}  \mid \mstr{\sigma} \mid \concat{\Sexp}{\Sexp} \mid
   \subst{\Sexp}{\Aexp}{\Aexp}\ \\
     \Ccomms\ni\Comm &::= \;\ppl{1}\skipc\ppl{2}\mid \ppl{1}\code{x:=e}\ppl{2}\mid 
     \ppl{1}\Comm;\ppl{2}\Comm\ppl{3}\mid \ppl{1}\code{if} ~({\tt b})~\{\ppl{2}\Comm\ppl{3}\}~\{\ppl{4}\Comm\ppl{5}\}\ppl{6} \\ & \mid\ 
     \ppl{1}\code{while} ~({\tt b})~ \{\ppl{2}\Comm\ppl{3}\}\ppl{4}\mid \ppl{1}\reflect(\Sexp)\ppl{2} \\
    \CommS\ni \P &::=\ppl{\iota}\Comm\sep\ppl{2} \qquad \mbox{ where }
    \mathsf{Id}\ni \code{x}  \mbox{ (Identifiers)},  n\in\integer, \sigma\in\Sigma^*
   \end{align*}}
   \caption{Syntax of $\CommS$}\label{synta}
\end{figure}

\noindent
In order to analyze a program $\prog\in\CommS$, we need to model it by building a corresponding control flow graph~\cite{compiler-design} ($\CFG$ for short), which embeds the control structure in the graph structure and leaves in the edges (or equivalently on the nodes) only the access to states, i.e., manipulation of the states (assignments) and guards. The approach we use is quite standard, and we follow \cite{compiler-design} for the construction of the control flow graph. For technical details see \cite{tops20}, here we show the construction on the example in Fig.~\ref{Figb:ifwhile}, where $i$ denotes the node corresponding to the program point $\ell_i$.
\begin{figure}[t]
\begin{subfigure}[b]{0.5\textwidth}
\begin{CenteredBox}	
\begin{lstlisting}[basicstyle=\fontsize{10}{10}\selectfont\ttfamily,escapeinside={(*}{*)}]
	(*$\pp{1}$*)x := 0;
	(*$\pp{2}$*)while (x<5)
		{(*$\pp{3}$*)x := x + 1(*$\pp{4}$*)};
	(*$\pp{5}$*)x:=7(*$\pp{6}$*)
	\end{lstlisting}
	\end{CenteredBox}
	\caption{}
	\end{subfigure}~
\begin{subfigure}[b]{0.5\textwidth}
\hspace{1cm}\includegraphics[scale=.11]{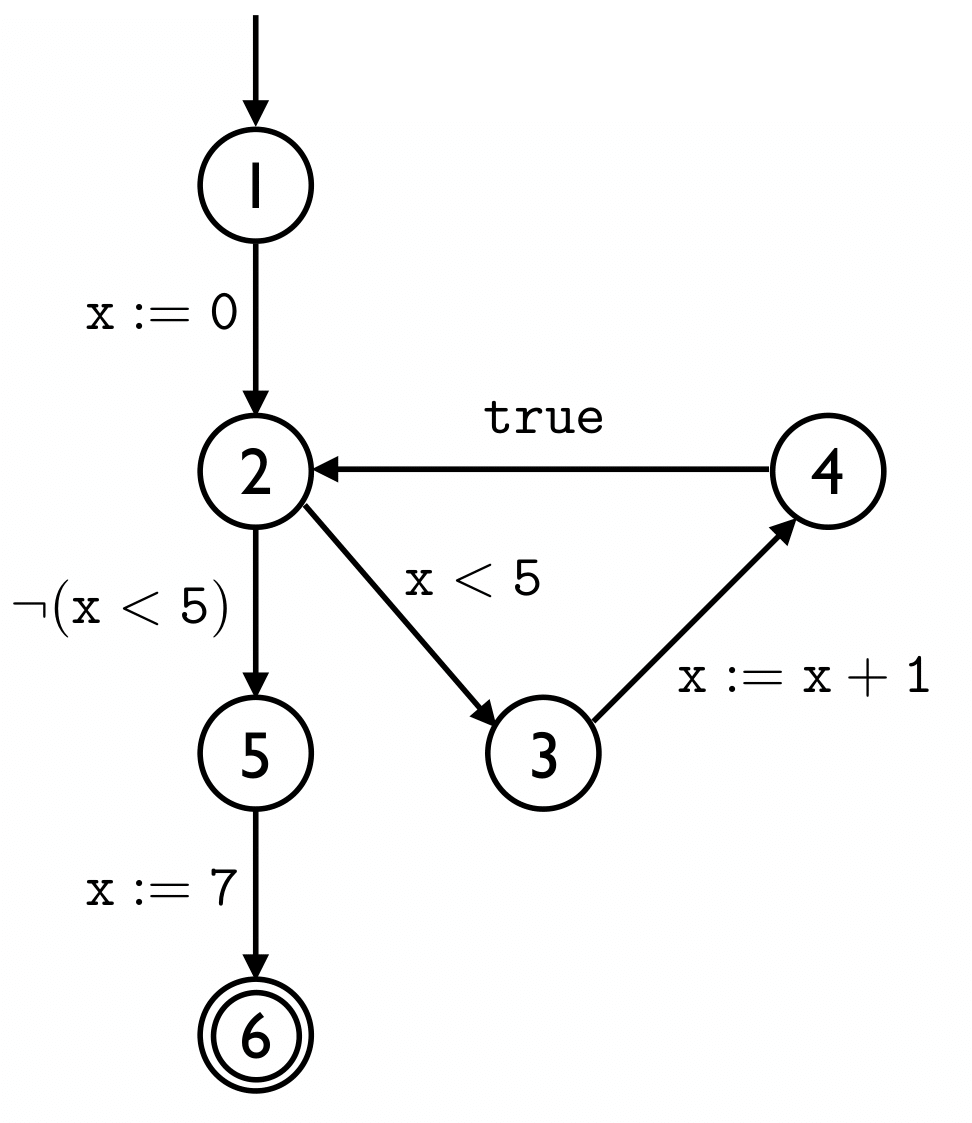}
\caption{}
\end{subfigure}
\caption{Example of $\CFG$: (a) Fragment of code and (b) corresponding $\CFG$.}\label{Figb:ifwhile}
\end{figure}
Note that, by construction \cite{tops20},  the language of the $\CFG$ edge labels is an intermediate language slightly different from the $\CommS$ grammar. Edge labels correspond to a primitive statement (i.e., an assignment or \code{eval}) or a boolean guard, namely they form the language $\CommG$ generated by the grammar $\tl ::= \; \code{x:=e}\mid {\tt b} \mid \reflect(\Sexp)$.
\paragraph*{Concrete Semantics.}
The concrete semantics of our language $\imp$ is intuitive and it is fully reported in~\cite{arceri2018}. Since
our aim is to analyze $\CommS$ programs by analyzing their $\CFG$s, we focus here only on the interpretation of $\CFG$'s labels~\cite{compiler-design}. In particular, we have to specify the semantics associated with each possible edge of the $\CFG$. In other words, we have to formalize how each statement transforms a current state, which is represented as a store, namely as an association between identifiers and values. It is well known that static program analysis works by computing (abstract) collecting semantics, namely, for each program point $\ell$ and for each variable $\code{x}$, it computes the set of values that the variable $\code{x}$ can have in any computation at the program point $\ell$. 
Hence, we define (collecting) memories $\mem$, associating with each variable a {\em set} of values. The basic values of $\CommS$ are integers, booleans and strings, hence
we define the set of memories as $\ok{\Mem\defi\Var \rightarrow (\wp(\mathbb{Z}) \cup \Bool \cup \wp(\Sigma^*))}$, ranged over the meta-variable $\mem$, where $\Bool = \wp(\{\false,\true\})$. Let us denote by $\cval$ this domain of collections of values $\wp(\mathbb{Z}) \cup \Bool \cup \wp(\Sigma^*)$.
The update of memory $\mem$ for a variable $\code{x}$ with set of values $v$ is denoted $\mem[\code{x}/v]$. The partial order $\sqsubseteq$ between memories is defined as $\mem_1 \sqsubseteq \mem_2 \Leftrightarrow \forall \code{x}\in\mathsf{Id} .\: \mem_1(\code{x}) \subseteq \mem_2(\code{x})$.
Finally, lub and glb of memories are computed point-wise, i.e., $\mem_1\sqcup \mem_2\defi\lambda \code{x} .\:\mem_1(\code{x})\cup\mem_2(\code{x})$ and $\mem_1\sqcap\mem_2\defi\lambda \code{x} .\:\mem_1(\code{x})\cap\mem_2(\code{x})$.\\ 
The collecting (input/output) semantics of statements $\Comm\in\CommG$ is defined as the function $\grasse{\Comm}: \Mem\rarr{} \Mem$. We denote by $\grasseb{\cdot}$ the collecting semantics of expressions, defined as additive lift\footnote{Let $f:S\ra S$ be a generic function, by {\em additive lift} we mean its extension to sets of elements, i.e., $\forall X\subseteq S$ we define $f(X)\defi\sset{f(x)}{x\in S}$. If $f:S\ra \wp(S)$, then its lift to sets of memories is $f(X)\defi\bigcup\sset{f(x)}{x\in S}$} to sets of memories of the standard expression semantics. We abuse notation by denoting as $\grasse{\cdot}$ also its additive lift to sets of statements.
\[
\begin{array}{rcl}
\grasse{\code{x:=e}}\mem&=&\mem[\code{x}/\grasseb{\Exp}\mem] \qquad 
\grasse{\Bexp}\mem=\mem\sqcap\bigsqcup\sset{\mem}{\grasseb{\Bexp}\mem=\true}\\
\grasse{\code{eval}(\Sexp)}\mem&=&\grasse{\grasseb{\Sexp}\mem\Cap\CommS}\mem
\end{array}
\]
where $\Cap$ is the intersection in the set of $\CommS$ programs.
By computing the traces of application of this transfer function, starting from any possible input memory, we precisely compute the maximal trace semantics \cite{mine2013}.

\paragraph*{Static Analysis on $\CFG$: Semantic Abstraction.}
It is well known that when we perform static analysis on a $\CFG$, we interpret, on the corresponding abstract domain, all the edges, and more specifically all the labels (in $\CommG$) \cite{compiler-design}. This is also a quite standard approach, but we recall it here for fixing the notation used.
%
We suppose to abstract values on the coalesced sum \cite{arceri2018} of the $\Sign$ abstract domain for integers, of the concrete domain for booleans and of the (deterministic) finite state automata abstract domain for strings \cite{arceri2018}\footnote{A string static analyzer using finite state automata abstract domain has been developed and it is available in~\cite{arceri2018}.}.
Let us consider an abstraction $\rho\in\uco(\cval)$\footnote{For the sake of simplicity here we abuse notation by considering a unique $\rho$ which is indeed the coalesced sum of three abstractions, one for integers, one for booleans and one for strings.} of the values manipulated by our language, we denote by $\rMem{\rho}:\Var\rarr{}\rho(\cval)$ the set of (collecting) memories, where sets of values are abstracted by $\rho$, ranged over $\rmem{\rho}$. 
In the following, we abuse notation by applying $\rho$ to memories in $\Mem$, simply by defining $\rho(\mem)\in\rMem{\rho}$ as $\rho(\mem):\code{x}\in\Var\mapsto \rho(\mem(\code{x}))$\footnote{For the sake of simplicity of presentation and implementation, we have considered here non-relational abstractions of data, anyway we believe that it is possible to easy extend our work to relational abstractions.}. In this way, we can see abstract memories as sets of concrete memories, and therefore as particular collecting memories, i.e., $\rMem{\rho}\subseteq\Mem$.
Finally, we can define the abstract edge effect $\grasse{\cdot}^\rho$ \cite{compiler-design} telling us how to abstractly interpret each edge of the $\CFG$: 
\[
\begin{array}{rcl}
\grasse{\code{x:=e}}^\rho\rmem{\rho}&=&\rmem{\rho}[\code{x}/\rho(\grasseb{\exp}^{\rho}\rmem{\rho})] \qquad\qquad
\grasse{\Bexp}^\rho\rmem{\rho}=\rmem{\rho}\sqcap\rho(\sqcup\sset{\mem}{\true\in\grasseb{\Bexp}^{\rho}\rmem{\rho}})\\
\grasse{\code{eval}(\Sexp)}^\rho\rmem{\rho}&=&\grasse{\grasseb{\Sexp}^\rho\rmem{\rho}\Cap\CommS}^\rho\rmem{\rho}
\end{array}
\]
where $\grasseb{\cdot}^\rho\defi\rho\comp\grasseb{\cdot}\comp\rho$.
The semantics of a path in the $\CFG$ is the composition of the interpretation of each edge, and the interpretation of an edge is the interpretation, given above, of its label \cite{compiler-design}.

This is clearly, what happens when the $\CFG$ is not abstracted, namely when the edge labels are single statements.
Finally, since we deal with potential abstract $\CFG$, we have to say how we execute them, potentially on an abstract semantics. The idea is simple, since we move from executing single statements to executing sets of statements, we simply take as execution of the abstract $\CFG$ the additive lift of the single statements executions. Since the semantics is always additive\footnote{A function is said to be {\em additive} if it commutes with least upper bound.}, in order to guarantee that everything works, also the semantic abstraction $\rho$ must be additive. Hence, in the following of the paper we always require $\rho$ to be additive.

\section{Semantic-driven Code Abstraction}\label{sec:asfa}\label{sect:cfg-as-sfa}
In this section, we study how we can model a syntactic abstraction of the $\CFG$ and which is its {\em relation} with the semantic abstraction, i.e., the code analysis.  

\paragraph*{Modeling Code Abstraction.}
%
Following the standard approach for abstracting objects, we should abstract each $\CFG$ in a set of $\CFG$s sharing an invariant property, i.e., an equivalence class of $\CFG$s. 
In particular, since we aim at abstracting code ($\CFG$) without changing the analysis performed on the code, we choose to abstract $\CFG$ by abstracting edge labels, and by leaving unchanged the control structure of the $\CFG$.  In other words, an abstract $\CFG$, denoted $\aCFG$, is a pair $\tuple{\nodes,\aedges}$, where we leave the nodes unchanged, while the edge labels are abstracted to sets of labels. Formally, $\aedges\subseteq\nodes\times\wp(\Psi)\times\nodes$, where $\Psi$ is the $\CFG$ label language.

Given $\eta\in\uco(\wp(\Psi))$, $\aetG\defi\tuple{\nodes(\tG),\edges^\eta(\tG)}$ is the $\aCFG$ built from a $\CFG$ $\tG$ in terms of $\eta$, where $\edges^\eta(\tG)\subseteq\nodes(\tG)\times\eta(\wp(\Psi))\times\nodes(\tG)$.

\begin{figure}
\centering
	\includegraphics*[scale=.14]{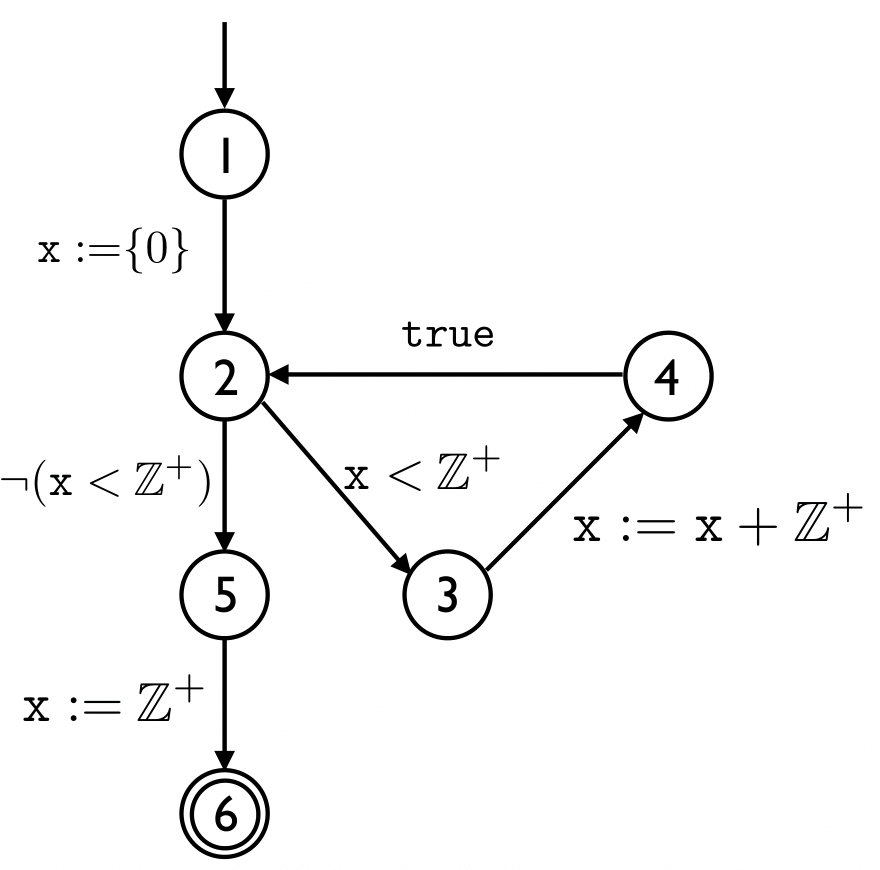}
\caption{$\CFG$ abstracted by signs.}	\label{fig:asfa}
\end{figure}
\noindent
As an example, consider the $\CFG$ in Fig.~\ref{Figb:ifwhile}, in Fig.~\ref{fig:asfa} we have the $\aCFG$ where numerical expressions are abstracted by $\sset{\code{m}}{m\in\Sign(n)}$\footnote{We use $n$ to denote the semantic value corresponding to the syntactic symbol $\code{n}$.} (where $\Sign$ is the well-known sign abstraction $\Sign\in uco(\wp(\Int))$ such as $\Sign(\wp(\Int))=\{\top,\cint^+,\cint^-,\{0\},\varnothing\}$). For instance, $\code{x:=x+1}$ is abstracted in $\code{x:=x+}\cint^+$ where $\code{x+}\cint^+\defi\sset{\code{x+n}}{n\in\cint^+}$, being $\Sign(1)=\cint^+$.
%
%

%
\paragraph*{Abstracting Code vs Abstracting Semantics.}
As previously noted, we aim at characterizing code abstractions, for dynamically generated code, for which the given analysis works precisely. 
Formally, let us consider the following equation:
\vspace{-.2cm}
\begin{equation}\label{eq:comFW}
\forall\rmem{\rho}\in\rMem{\rho}\subseteq\Mem.\:\forall\varphi\in\Psi.\:\interp{\eta(\varphi)}\rmem{\rho}=\interp{\eta(\varphi)}^\rho\rmem{\rho}
\vspace{-.2cm}
\end{equation}
If this equality does not hold it means that the abstract semantic interpretation $\interp{\cdot}^\rho$ merges predicates distinguished by $\eta$. Namely, when the program is observed by means of its (abstract) semantics the actual abstraction of predicates is not precisely $\eta$, but it is $\eta$ affected in some way by $\interp{\cdot}^\rho$. By changing the point of view, we have that, in this case, the analysis cannot precisely interpret the abstract code, since $\eta$ abstracts the code by distinguishing information that $\rho$ cannot distinguish.\\
As an example, consider the sign domain above, when $\eta(\code{x:=5})=\sset{\code{x:=n}}{1\leq n\leq 5}$ the equation does not hold since the concrete semantics of this set does not take {\em any} positive value for $\code{x}$. While, if $\eta(\code{x:=5})=\sset{\code{x:=n}}{n\in\cint^+\cup\{0\}}$, then Eq.~\ref{eq:comFW} holds since its concrete semantics is precisely the set of  non-negative values.
It is worth noting that Eq.~\ref{eq:comFW} is a forward completeness \cite{GQ01} of the code abstraction w.r.t.\ the semantic interpretation, meaning that the semantic abstraction does not add imprecision to the code one. \\
%
In order to investigate the relation existing between the code abstraction $\eta$ and the semantic abstraction $\rho$, we observe that, whenever we have a semantic abstraction $\rho$, we have a natural code abstraction induced by $\rho$. Namely, by only observing (abstract) information about the computation, we cannot distinguish statements with the same (abstract) semantics, independently from what any possible code abstraction does. For instance, if we analyze parity of program variables, we are unable to distinguish $\code{x:=2}$ from $\code{x:=4}$, independently from how a potential code abstraction $\eta$ is defined on $\code{x:=2}$.  The first step consists in defining a code abstraction for expressions in terms of semantic one. Consider $\rho\in\uco(\cval)$, we define $\heta(\exp)$ inductively on the expressions structure
\[
\begin{array}{ll}
\heta(\Aexp)&:
\left \{
\begin{array}{ll}
\heta(\Aexp_1 \mbox{\tt op}\:\Aexp_2)\defi \sset{\Aexp'\mbox{\tt op}\:\Aexp''}{\Aexp'\in\heta(\Aexp_1),\Aexp''\in\heta(\Aexp_2)}\defi\heta(\Aexp_1)\mbox{\tt op}\:\heta(\Aexp_2)\\
\heta(\code{x})\defi \code{x},\qquad\heta(\code{n})\defi \sset{\code{m}}{m\in\rho(\{n\})}\\
\end{array}
\right.\\
\heta(\Bexp)&:
\left \{
\begin{array}{ll}
\heta(\Bexp_1 \mbox{\tt bop}\:\Bexp_2)\defi  \heta(\Bexp_1)\mbox{\tt bop}\:\heta(\Bexp_2),\qquad
\heta(\neg\Bexp)\defi\neg\heta(\Bexp)\\
\heta(\code{x})\defi \code{x},\qquad \heta(\true)\defi\sset{\code{t}}{t\in\rho(\true)},\qquad \heta(\false)\defi\sset{\code{t}}{t\in\rho(\false)}
\end{array}
\right. \\
\heta(\Sexp)&:
\left \{
\begin{array}{ll}
\heta(\concat{\Sexp_1}{\Sexp_2})\defi  \concat{\heta(\Sexp_1)}{\heta(\Sexp_2)},\\
\heta(\subst{\Sexp}{\Aexp_1}{\Aexp_2})\defi\subst{\heta(\Sexp)}{\heta(\Aexp_1)}{\heta(\Aexp_2)}\\
\heta(\code{x})\defi \code{x},\qquad\heta(\mstr{\sigma})\defi\sset{\mstr{\delta}}{\delta\in\rho(\sigma)}
\end{array}
\right. 
\end{array}
\]
At this point, we can characterize the $\CFG$ labels abstraction $\ov{\CAbs}[\rho]:\wp(\Psi)\lra\wp(\Psi)$, as the additive lift of the function
%
%
%
\[
\begin{array}{rl}
\ov{\CAbs}[\rho](\code{x:=}\exp)&\defi x:=\heta(\Exp)\defi\sset{x:=\Exp'}{\Exp'\in\heta(\Exp)}\\ 
\ov{\CAbs}[\rho](\bexp)&\defi \heta(\bexp)\qquad \qquad
\ov{\CAbs}[\rho](\code{eval}(\Sexp))\defi \code{eval}(\heta(\Sexp))
\end{array}
\]
where $\code{eval}(\heta(\Sexp))$ is treated as the implicit representation of all the statements that it can execute, namely it represents the (potentially infinite) set $\sset{\tc}{\grasse{\tc}\mem\sqsubseteq\grasse{\grasseb{\Sexp}^\rho\Cap\CommS}^\rho\mem}$.

\noindent
The following result is immediate by construction.
\begin{proposition}
Given $\rho\in\uco(\cval)$, then $\ov{\CAbs}[\rho]\in\uco(\wp(\Psi))$ and it is additive.
\end{proposition}

Finally, in order to show that this code abstraction can be used to force satisfiability of  Eq.~\ref{eq:comFW}, we have first to characterize the meaning of interpreting an edge label abstracted by $\ov{\CAbs}[\rho]$:
\[
\begin{array}{rcl}
\grasse{\code{x:=}\heta(\exp)}\mem&=&\bigsqcup\sset{\grasse{\code{x:=e}'}\mem}{\exp'\in\heta(\exp)} \qquad\qquad
\grasse{\heta(\Bexp)}\mem=\bigsqcup\sset{\grasse{\bexp'}\mem}{\bexp'\in\heta(\bexp)}\\
\grasse{\code{eval}(\heta(\Sexp))}\mem&=&\bigsqcup\sset{\grasse{\tc}\mem}{\grasse{\tc}\mem\sqsubseteq\grasse{\grasseb{\Sexp}^\rho\Cap\CommS}^\rho\mem}
\end{array}
\]

\noindent
Then we have the following results

\begin{lemma}\label{lemmaImp}
Given $\rho\in\uco(\cval)$ additive, then $\forall\exp.\:\forall\mem\in\rMem{\rho}.\:\grasseb{\heta(\exp)}\mem=\grasseb{\exp}^\rho\mem$ (trivially implying $\exp'\in\heta(\exp)\ \Lra\ \forall\mem\in\rMem{\rho}.\:\grasseb{\exp'}\mem\subseteq\grasseb{\exp}^\rho\mem$) and $\forall\Phi\in\wp(\Psi).\:\forall\mem\in\rMem{\rho}.\:\interp{\ov{\CAbs}[\rho](\Phi)}\mem=\interp{\Phi}^\rho\mem$.
\end{lemma}
\begin{proof}
Let us prove first the property for expressions by induction on the syntactic structure of $\exp$.
\begin{itemize}
\item[ ] $\exp=\code{n}$: $\grasseb{\heta(\exp)}\mem=\grasseb{\heta(\code{n})}\mem\defi\rho(n)$, while $\grasseb{\exp}^\rho\mem=\grasseb{\code{n}}^\rho\mem=\rho(n)$ (where $\grasseb{\code{n}}\mem=n$);
\item[ ] $\exp=\code{x}$: $\grasseb{\heta(\exp)}\mem=\grasseb{\heta(\code{x})}\mem\defi\grasseb{\code{x}}\mem=\mem(\code{x})$, while $\grasseb{\exp}^\rho\mem=\grasseb{\code{x}}^\rho\mem=\rho(\mem(\code{x}))=\mem(\code{x})$ (since $\mem\in\rMem{\rho}$);
\item[ ] $\exp=\exp_1\:\code{op}\:\exp_2$: Suppose $\code{op}$ any arithmetic or boolean operator.\\ $\grasseb{\heta(\exp)}\mem=\grasseb{\heta(\exp_1\:\code{op}\:\exp_2)}\mem\defi\grasseb{\heta(\exp_1)\:\code{op}\:\heta(\exp_2)}\mem=\grasseb{\heta(\exp_1)}\mem\:\code{op}\:\grasseb{\heta(\exp_2)}\mem=\grasseb{\exp_1}^\rho\mem\:\code{op}\:\grasseb{\exp_2}^\rho\mem$ by inductive hypothesis. But this is precisely $\grasseb{\exp_1\:\code{op}\:\exp_2}^\rho\mem$ since $\code{op}$ is computed on the semantics as additive lift to sets.
\item[ ] Analogously, we can prove all the other cases.
\end{itemize}
Now, let us prove the fact for $\CFG$ single edge labels, again by induction on the syntactic structure. Note that, being $\rho$ additive then also $\grasse{\cdot}^\rho$ is additive, being also the concrete semantics additive on sets of statements.
\begin{itemize}
\item[ ]
\[
\begin{array}{lll}
\grasse{\ov{\CAbs}[\rho](\code{x:=e})}\mem&=&\grasse{\code{x:=}\heta(\exp)}\mem\\
&=&\bigsqcup\sset{\grasse{\code{x:=e}'}\mem}{\exp'\in\heta(\exp)}\\
&=&\bigsqcup\sset{\mem[\code{x}/\grasseb{\exp'}\mem]}{\exp'\in\heta(\exp)}\\
&=&\mem[\code{x}/\bigcup\sset{\grasseb{\exp'}\mem}{\exp'\in\heta(\exp)}]\\
&=&\mem[\code{x}/\bigcup\sset{\grasseb{\exp'}\mem}{\grasseb{\exp'}\mem\subseteq\grasseb{\exp}^\rho\mem}]\\
&=&\mem[\code{x}/\grasseb{\exp'}^\rho\mem]=\grasse{\code{x:=}\exp}^\rho\mem\\
\end{array}
\]
\item[ ] 
\[
\begin{array}{lll}
\grasse{\ov{\CAbs}[\rho](\bexp)}\mem&=&\grasse{\heta(\bexp)}\mem\\
&=&\bigsqcup\sset{\grasse{\bexp'}\mem}{\bexp'\in\heta(\bexp)}\\
&=&\bigsqcup\sset{\mem\sqcap\bigsqcup\sset{\mem}{\grasseb{\Bexp'}\mem=\true}]}{\bexp'\in\heta(\bexp)}\\
&=&\mem\sqcap\bigsqcup\sset{\mem}{\grasseb{\Bexp'}\mem=\true,\ \bexp'\in\heta(\bexp)}\\
&=&\mem\sqcap\bigsqcup\sset{\mem}{\grasseb{\Bexp'}\mem=\true,\ \grasseb{\bexp'}\mem\subseteq\grasseb{\bexp}^\rho\mem}\\
&=&\mem\sqcap\bigsqcup\sset{\mem}{\true\in\grasseb{\bexp}^\rho\mem}=\grasse{\bexp}^\rho\mem\\
\end{array}
\]
\item[ ] 
\[
\begin{array}{lll}
\grasse{\ov{\CAbs}[\rho](\code{eval}(\Sexp))}\mem&=&\grasse{\code{eval}(\heta(\Sexp))}\mem\\
&=&\bigsqcup\sset{\grasse{\tc}\mem}{\grasse{\tc}\mem\sqsubseteq\grasse{\grasseb{\Sexp}^\rho\Cap\CommS}^\rho\mem}\\
\mbox{By additivity of $\grasse{\cdot}^\rho$}&=&\grasse{\grasseb{\Sexp}^\rho\Cap\CommS}^\rho\mem=\grasse{\code{eval}(\Sexp)}^\rho\mem\\
\end{array}
\]
\end{itemize}
Finally, for each set of labels $\Phi$, we have that $\grasse{\ov{\CAbs}[\rho](\Phi)}\mem=\bigsqcup_{\varphi\in\Phi}\grasse{\ov{\CAbs}[\rho](\varphi)}\mem=\bigsqcup_{\varphi\in\Phi}\grasse{\varphi}^\rho\mem=\grasse{\Phi}^\rho\mem$, since all the involved functions are additive by definition or by construction.
\end{proof}

\noindent
Then we have that:
\begin{theorem}\label{th:omega}
Let $\rho\in\uco(\cval)$ additive, and $\eta\in\uco(\wp(\Psi))$. Then $\ov{\eta}_\uparrow\defi\ov{\CAbs}[\rho]\circ\eta$ satisfies Eq.~\ref{eq:comFW}.
\end{theorem}
\begin{proof}
It is worth noting that, we trivially have by abstraction that $\forall\varphi\in\Psi.\:\interp{\eta_\uparrow(\varphi)}\subseteq\interp{\eta_\uparrow(\varphi)}^\rho$. Let us prove the other implication: $\forall\varphi\in\Psi$
\[
\begin{array}{llr}
\interp{\eta_\uparrow(\varphi)}&=\interp{\ov{\CAbs}[\rho]\circ\eta(\varphi)}&\\
&=\interp{\ov{\CAbs}[\rho]\circ\ov{\CAbs}[\rho]\circ\eta(\varphi)}&\qquad [\mbox{By properties of uco}]\\
&=\interp{\ov{\CAbs}[\rho]\circ\eta(\varphi)}^\rho& [\mbox{By Lemma.~\ref{lemmaImp}}]\\
&=\interp{\eta_\uparrow(\varphi)}^\rho
\end{array}
\]
\vspace{-.5cm}
\end{proof}

\noindent
This result tells us that by taking a code abstraction more abstract than (or equal to) $\ov{\CAbs}[\rho]$,  we guarantee that the abstract interpretation $\rho$ {\em can be} performed on the abstracted program (Eq.~\ref{eq:comFW}). We have so far proved that it is always possible to force  Eq.~\ref{eq:comFW}, in order to make it possible to continue the analysis (observing $\rho$) also on the abstracted code. In the following we show how this framework can be integrated with the existing analysis of dynamic code \cite{tops20} in order to improve its precision.

\section{An Improved Dynamic Code Analysis}\label{sec:appl}
In this section we show how the constructive code abstraction characterization, provided in the previous section, can be used for representing the code approximation which soundly captures the potential code executed by a string-to-code statement. As we will show, without abstracting code, we cannot capture situations where the collecting semantics on strings generates sets of statements that cannot be represented by using the concrete syntax. Nevertheless, we must also observe that the analyzer cannot change dynamically with the generated code, hence the abstraction {\em must} be driven by the semantic property analyzed. This means that, without using the proposed framework, the analysis would surely be less precise in those situations where code abstraction becomes a necessity. \\

\noindent
Let us summarize how we propose to exploit the framework:
\begin{itemize}
\item[$\bullet$] Consider a fixed semantic abstraction $\rho\in\uco(\cval)$ and a corresponding static analyzer, designed in such a way that it can interpret also code abstracted by $\ov{\CAbs}[\rho]$. 
\item[$\bullet$] Analyze the program, and when an $\code{eval}$ is met, extract the language of its argument. If the language is infinite (under specific conditions that we will discuss) build the abstract $\CFG$ approximating it and extract the corresponding code abstraction $\eta$. In general, this code abstraction $\eta$ is not more abstract than $\ov{\CAbs}[\rho]$ (the code abstraction already embedded in the static analyzer, depending only on $\rho$);
\item[$\bullet$] Build $\ov{\CAbs}[\rho]\circ\eta$ in order to make also the generated code (approximated by the generated abstract $\CFG$) analyzable by the static analysis for $\rho$.
\end{itemize}
%
\paragraph*{Analyzing Dynamic Code.}
Let $\rho$ be a static analysis performing in particular $\rhos\in uco(\wp(\cstr))$ on strings, where $\cstr = \alphabet^*$ denotes strings over a finite alphabet $\alphabet$. 
Note that, our analyzer has to work on any (abstract) $\CFG$ that can be dynamically generated, hence it has to be designed with this purpose in mind. In particular, as we will show, we will generate only abstract $\CFG$s with a code abstraction $\eta$ complete w.r.t.\ $\rho$. This means, by construction, that $\eta$ must be more abstract than $\ov{\CAbs}[\rho]$, which means that each set of elements in $\eta$ corresponds to a subset of the elements (abstract predicates) of $\ov{\CAbs}[\rho]$. Hence, in order to guarantee to interpret predicates in any $\eta$ complete, it is sufficient to design the analyzer soundly interpreting  any abstract predicate in $\ov{\CAbs}[\rho]$. For instance, $\ov{\CAbs}[\Sign]$ is the abstraction containing all the predicates, involving integers, of the form $\code{x:=S}$, $\code{x<S}$, etc, with $\code{S}\in\Sign$, e.g., an abstract predicate is $\code{x:=}\mathbb{Z}^+$, and the analyzer for $\Sign$ should be able to interpret also such abstract predicates. \\
Let \code{x} be the input string parameter of an \code{eval} statement, we denote by $\sval{\rhos}{\code{x}}$ the abstract value for $\code{x}$ computed by the analysis on $\rhos$.
For example, suppose that the collection of values for the string $\code{x}$ before the \code{eval} is $\{\code{a:=0}, \code{a:=1}\}$.
By defining $\rhos$ as the $k$-bounded string set abstract domain~\cite{amadini18}, with $k = 2$, $\sval{\rhos}{\code{x}} = \{\code{a:=0}, \code{a:=1}\}$, while by using the prefix abstract domain $\overline{\mathcal{PR}}$ ~\cite{costantini15}, $\sval{\mbox{\tiny $\overline{\mathcal{PR}}$}}{\code{x}} = \sset{\code{a:=s}}{s \in \cstr}$. When the abstracted string and the abstraction is clear from the context, we simply denote this set by $\cS$ and we assume (for the sake of simplicity) that any string in $\cS$ is an executable language statement\footnote{Note that, this assumption corresponds to a decidable condition, hence it is possible to check it and to implement ad hoc solutions when it does not hold.}. In the following, we abuse notation by denoting $\cS$ also the automaton recognizing the language. \\
%
Consider for example, the  program reported in Fig.~\ref{fig:p-abs-left}, a program building and manipulating the string \code{str} at run-time, which is, afterwards, interpreted as executable code, being the input parameter of the string-to-code statement $\code{eval}$.
Since the value of $\code{N}$ is unknown at compile-time, we cannot predict the precise number of iterations of the \code{while}-loop. In this case, a suitable string abstract analysis would approximate the value of $\code{str}$, before the \code{eval} execution, to an abstract value corresponding to an over-approximation of the possible values for $\code{str}$, which may be also, due to abstraction, an infinite set of strings, and therefore an infinite set of possible programs.
\begin{figure}[t]
	\begin{subfigure}[b]{0.5\textwidth}
		\begin{CenteredBox}
			\begin{lstlisting}[basicstyle=\fontsize{10}{10}\selectfont\ttfamily,escapeinside={(*}{*)}]
			(*$\pp{1}$*)str := "x:=5"; (*$\pp{2}$*)i := 0;
			(*$\pp{3}$*)while (i < N) {
				(*$\pp{4}$*)str := str + "5";
				(*$\pp{5}$*)i:=i+1;(*$\pp{6}$*)
			}
			(*$\pp{7}$*)str := "if(x<5){"  + str 
					+ "}else{x:=1};";
			(*$\pp{8}$*)eval(str)(*$\pp{9}$*)
			\end{lstlisting}
		\end{CenteredBox}
		\caption{}
		\label{fig:p-abs-left}	
	\end{subfigure}
	~
	\begin{subfigure}[b]{0.5\textwidth}
		\centering
		\includegraphics[scale=0.11]{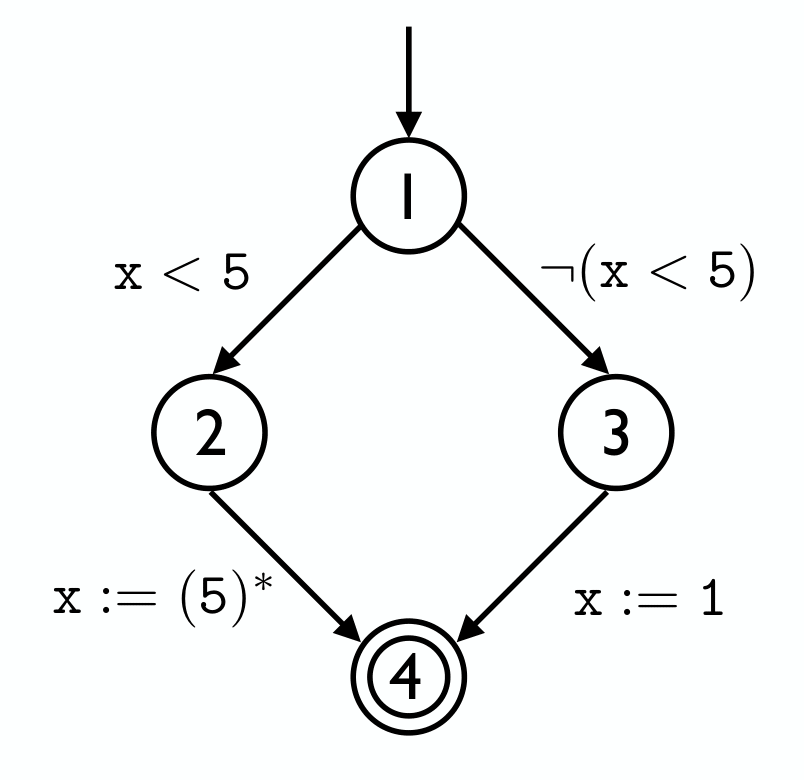}
		\caption{}
		\label{fig:p-abs-right}	
	\end{subfigure}
	\label{fig:p-abs}
	\caption{(a) Dynamically-generating code sample. (b) $\CFG$ associated to \code{str} labeled with abstract expressions.}
\end{figure}
For instance, in the example, if we abstract strings into the regular expression abstract domain~\cite{choi2006} (or equivalently into the finite state automata abstract domain~\cite{arceri2018}), the value of \code{str} after the \code{while} loop will be the abstract value $\mathtt{x :=5(5)^*;}$ corresponding to an infinite set of programs, i.e., \code{x:=5;}, \code{x:=55}, \code{x:=555;}\dots.
In this case, the common practice for analyzing \code{eval} is simply to give up with the analysis, for example by halting the analysis throwing an exception \cite{jensen2012} or forbidding its usage \cite{jsai}.

Let $\rhoimp$ be the abstract domain for all the possible values (integers, strings and booleans) \cite{tops20}. Note that, $\ov{\CAbs}[\rhoimp]$ contains, for integers,  predicates like the ones in the abstract $\CFG$ in Fig.~\ref{fig:asfa}. \\
The analysis $\rhoimp$ at point $\ell_3$, due to widening\footnote{Widening is a fix-point accelerator used in infinite domains with infinite ascending chains, namely where the semantic fix-point computation may diverge. In this case we use a widening on automata defined in \cite{choi2006}} applied in the analysis of the while loop \cite{arceri2018}, abstracts the value of \code{str} in the infinite language $\sset{\code{x:=s}}{s\in (5)^+}$ (namely \code{x} is assigned to any value represented by a finite sequence of $5$). Hence, at point $\ell_8$ the analysis abstracts \code{str} to the strings set  
$
\cS_{\mathtt{str}} = \sset{\code{if(x<5)\{x:=s\}else\{x:=1\}}}{s \in (5)^+}
$ meaning that, the true-branch of the string that may be transformed by $\code{eval}$ may be either $\code{x:=5}$, or $\code{x:=55}$, or $\code{x:=555}, \dots$. The automaton corresponding to the abstract value of $\code{str}$ is reported in Fig.~\ref{fig:fa}, and it denotes an infinite language, i.e., an infinite set of possible statements. 
Unfortunately, this is a problem for the analysis provided in \cite{tops20}, where the language containing all the possible strings would be returned, losing any precision.

\begin{figure}[t]
	\centering
	\includegraphics[scale=0.55]{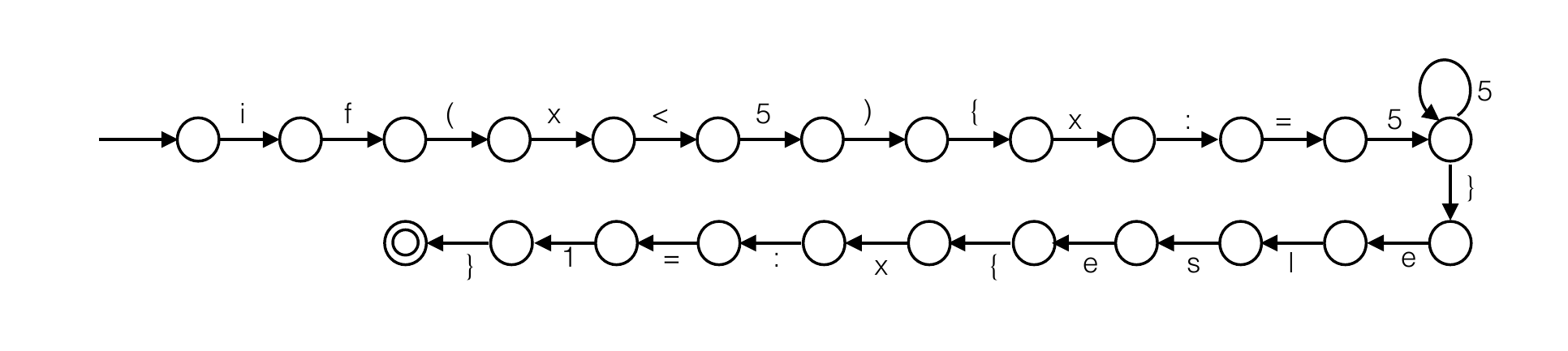}
	\caption{Finite state automaton corresponding to the abstract value of $\code{str}$.}\label{fig:fa}
\end{figure}

\paragraph*{Generating the Code: From Automata to $\CFG$s.}
At this point, we have the (potentially infinite) language of the $\reflect$ argument (and hence an automaton $\cS$), and the goal is to generate a $\CFG$ modeling an over-approximation of the executable code contained in the language of the automaton $\cS$. The idea is to generate a $\CFG$ from a language of strings, i.e., from an automaton, by performing a parsing on the paths of the automaton. Indeed, we have defined and implemented an algorithm\footnote{In the following, we only discuss the main parts of the algorithm for space limitations.}, reported in Alg.~\ref{alg:autsfa}, performing an abstract parser on automata that, given an automaton $\cS$, returns the $\CFG$ $\cP$ that over-approximates, for each $s \in \cS$ (executable), the concrete execution of $\code{eval}$.

The idea of  Alg.~\ref{alg:autsfa} is to perform a depth-first search on the automaton and, when a language statement 
is recognized, to generate an edge in the $\CFG$. This phase is handled by lines 3-13 of Alg.~\ref{alg:autsfa}, building the set of nodes $\nodes$ and the set of edges $\edges$ of the resulting $\CFG$ $\cP$. The set $W$ contains the states of the finite state automaton for which we still have to generate edges in the $\CFG$ and it is initialized, at line 2, with the initial state $q_0$. 
At this point, Alg.~\ref{alg:autsfa} looks for language statements readable from any path of the input automaton starting from a state $q$, taken from $W$, by means of the module $\mathrm{ReduceStmts}$ (line 5). In particular, $\mathrm{ReduceStmts}$  returns a set of triples $(q', \Comm, q'')$, where each returned triple means that from $q'\in Q$ to $q''\in Q$ a language statement $\Comm$ has been recognized. 
\begin{figure}[t]
	\scalebox{0.8}{
	\begin{algorithm}[H]
		\KwData{$\cS = (Q, \alphabet, \delta, q_0, F)$}
		\KwResult{$\CFG$ $\cP$ over-approximating executable strings of $\cS$}
		
		$\cS = \mathrm{ReduceCycles}(\cS)$\;
		$\nodes \leftarrow \varnothing$;  $\edges\leftarrow\varnothing$;
		$W \leftarrow \{q_0\}$;
		$visited \leftarrow \varnothing$;
		
		\While{$W \neq \varnothing$}{
			select and remove $q$ from $W$\;
			$stmts \leftarrow \mathrm{ReduceStmts}(\cS, q)$\;
			\ForEach{$(q', \Comm, q'') \in stmts$}{
				$\nodes \leftarrow \nodes' \cup \{\lab(q'),\lab(q'')\}$\;
				
				$\edges \leftarrow \edges \cup \{(\lab(q'), \Comm, \lab(q''))\}$\;
				
				$visited \leftarrow visited \cup \{q'\}$\;
				$W \leftarrow W \cup \{q''\}$\;
				$W \leftarrow W \smallsetminus visited$\;
			}
		}
		\textbf{return} $\cP = \tuple{\nodes,\edges}$\;
		\caption{}
		\label{alg:autsfa}
	\end{algorithm}
}
\end{figure}
The set returned by $\mathrm{ReduceStmts}$ corresponds to the set of statements of $\cP$ readable from the state $q$, hence they are added to $\edges$, substituting the reached states with the corresponding labels by means of the function $\lab$ (lines 7-8). At this point, we need to look for the statements that can be read from $q''$, hence, $q''$ is added to $W$ in order to be eventually processed at the next iterations of the while loop at lines 3-13. When there are no more states of $\cS$ to be processed, namely when $W$ is empty, the $\CFG$ $\cP = \tuple{\nodes,\edges}$ is returned (line 14), with entry label $\lab(q_0)$ and exit labels the ones associated with the states in $F$.

Problems arise when the automaton contains cycles (namely, when the automaton denotes an infinite language). In this case, Alg.~\ref{alg:autsfa} first transforms, at line 1, the input automaton, over the alphabet $\alphabet$, in an automaton without cycles, over the alphabet $\alphabet \cup \wp(\alphabet^*)$, by means of the module $\mathrm{ReduceCycles}$. 
Given an input automaton $\cS$, we retrieve the cycles of $\cS$ using the well-known Tarjan's algorithm~\cite{tarjan} for identifying cycles. Then, for each detected cycle of $\cS$, we check whether the string read by the cycle is a whole statement $\mathsf{r}$ or not. In the first case, we substitute the cycle of the string $\mathsf{r}$ in the automaton, i.e., $\mathsf{r}^*$, with the automaton reading the string corresponding to the statement \code{while(true)\{} $\mathsf{r}$ \code{\}} over the alphabet $\alphabet$.
Otherwise, if the cycle does not read a whole statement, the idea is to collapse the cycle in a single transition, labeled with the regular expression corresponding to what is read in the cycle, i.e., denoting a set of string on $\alphabet$ ($\wp(\alphabet^*)$). Hence the resulting automaton is on the alphabet $\alphabet \cup \wp(\alphabet^*)$. In Fig.~\ref{fig:freestar} we report an example of application of $\mathrm{ReduceCycles}$ algorithm.
\begin{figure}[t]
	\begin{subfigure}[b]{0.5\textwidth}
		\centering
		\includegraphics[scale=0.55]{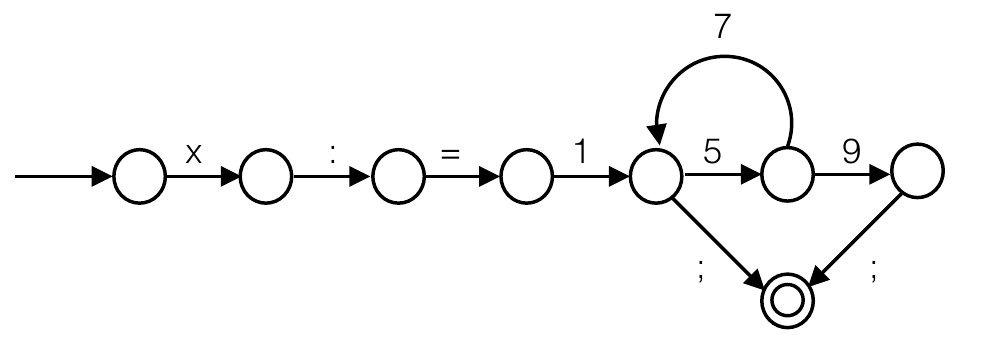}
		\caption{}
		\label{fig:freestar-a}
	\end{subfigure}
	\begin{subfigure}[b]{0.5\textwidth}
		\centering
		\includegraphics[scale=0.27]{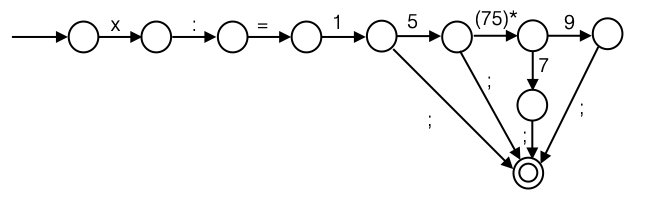}
		\caption{}
		\label{fig:freestar-b}
	\end{subfigure}
	\caption{(a) Finite state automaton with cycle. (b) Result of $\mathrm{ReduceCycles}$.}
	\label{fig:freestar}
\end{figure}
%
%
%
As example note that, by applying Alg.~\ref{alg:autsfa} to the automaton for $\cS_{\mathtt{str}}$ in Fig.~\ref{fig:fa}, we generate the $\CFG$ $\cP_{\mathtt{str}}$, depicted in Fig.~\ref{fig:p-abs-right}. 
%
%
It is worth noting that the $\CFG$ obtained so far may contain abstract expressions on edges, hence edges may represent an infinite collection of statements. At this point, we need to approximate these edges for making it possible to analyze the $\CFG$. 

%
%
\paragraph*{Making the Code Analyzable: Abstracting the $\CFG$.}
Let us recall that we have to perform the analysis $\rho$ also on the resulting code, in order to continue the static analysis. Hence, as observed before, we have to combine the code abstraction corresponding to the generated (abstract) $\CFG$ with the code abstraction induced by the semantic abstraction $\rho$, i.e., $\ov{\CAbs}[\rho]$, which models, as code abstraction, the analysis.\\
First of all, we have to formally characterize the abstraction $\eta$ induced by the construction of the $\CFG$ given above, namely we characterize how the construction abstracts together different predicates.
Let us build a code abstraction starting from the $\CFG$ $\cP=\tuple{\nodes,\edges}$ built in Alg.~\ref{alg:autsfa}: In particular, let $\mbox{\sl Merge}\defi\sset{\sset{\varphi\in\Psi}{\tuple{\ell',\varphi,\ell''}\in\edges}}{\ell',\ell''\in \nodes}\subseteq\Psi$ be the set of collections of predicates between any pair of states in the $\CFG$, we define 
\begin{equation}\label{eq:eta-As}
\eta^{\cP}(\wp(\Psi))\defi\wp(\sset{X\in \mbox{\sl Merge}}{\forall Y\in \mbox{\sl Merge}\smallsetminus\{X\}.\:X\cap Y=\varnothing})\in\uco(\wp(\Psi))
\end{equation}
Note that, this abstraction, being characterized starting from the $\CFG$ is defined only in terms of a finite subset of $\Psi$, namely on the predicates in the given $\CFG$, i.e.,  $\Psi^{\cP}\defi\Psi\cap\sset{\varphi}{\tuple{\ell',\varphi,\ell''}\in\edges}$. \\
In the example, $\ok{\Psi^{\cP_{\mathtt{str}}(\wp(\Psi))}=\{\sset{\code{x:=s}}{s\in (5)^+},\{\code{x:=1}\},\{\code{(x<5)}\},\{\neg\code{(x<5)}\}\}}$, hence we have that  $\ok{\eta^{\cP_{\mathtt{str}}}=\wp(\Psi^{\cP_{\mathtt{str}}})}$, being $\ok{\Psi^{\cP_{\mathtt{str}}}}$ already a partition. In Fig.~\ref{fig:eta-eval} this abstraction is partially depicted.\\
\begin{figure}[t]
	\begin{subfigure}{1\textwidth}
		\centering
		\begin{tikzpicture}[scale=0.3]
		
		\node (top) at (0,8) {$\top$};
		\node (e) at (-4,5) {$\{\code{x:=1}\} \cup \sset{\code{x:=s}}{s \in (5)^+}$};
		\node (em) at (7.5,5) {$\vee-$closure};
		\node (d) at (13,2) {$\{\neg\code{(x<5)}\}$};
		\node (c) at (4,2) {$\{\code{(x<5)}\}$};
		\node (b) at (-4,2) {$\sset{\code{x:=s}}{s \in (5)^+}$};
		\node (a) at (-13,2) {$\{\code{x:=1}\}$};
		\node (bot) at (0,-1) {$\bot$};
		
		\draw (b) -- (bot) --(a)--(e) -- (b);
		\draw[loosely dotted] (em) -- (top)--(e);
		\draw[loosely dotted] (em) -- (c);
		\draw (c)--(bot) -- (d);
		\draw[loosely dotted] (d)--(em) -- (b);
		\end{tikzpicture}
		\caption{}
		\label{fig:eta-eval}
	\end{subfigure}
~
	\begin{subfigure}{1\textwidth}
	\centering
	\begin{tikzpicture}[scale=0.3]
	\node (top) at (0,10) {$\top$};
	\node[align=center] (d) at (-11, 7) {$\code{x:=} \mathbb{Z}^+ \vee (\code{x<}\mathbb{Z}^+)$};
	\node[align=center] (e) at (11,7) {$\code{x<}\mathbb{Z}^+ \vee \neg\code{(x<}\mathbb{Z}^+\code{)}$};
	\node[align=center] (f) at (0, 7) {$\code{x:=}\mathbb{Z}^+ \vee \neg\code{(x<}\mathbb{Z}^+\code{)}$};
	
	\node[align=center]  (c) at (7,4) {$\neg\code{(x<}\mathbb{Z}^+\code{)}$};
	\node[align=center] (b) at (0,4) {$\code{(x<}\mathbb{Z}^+\code{)}$};
	\node[align=center] (a) at (-7,4) {$\code{x:=}\mathbb{Z}^+$};
	\node[align=center] (bot) at (0,1) {$\bot$};
	
	\draw  (bot) -- (c) -- (e)-- (top) --(d) -- (b) -- (bot) -- (a) -- (d);
	\draw (top) -- (f) -- (a);
	\draw (c) -- (f);
	\draw (e) -- (b);
	
	\end{tikzpicture}
	\caption{}
	\label{fig:omega-eval}
\end{subfigure}
	\caption{(a) Code abstraction $\eta^{\cP_{\mathtt{str}}}$ w.r.t. the $\CFG$ reported in Fig.~\ref{fig:p-abs-right}, (b) Code abstraction $\ov{\CAbs}[\mathsf{\rhoimp}]^{\cP_{\mathtt{str}}}$ }
\end{figure}
Finally, we need to satisfy Eq.~\ref{eq:comFW} (completeness) between the code abstraction $\eta^{\cP}$, built so far, and the static analysis, modeled as a semantic abstraction $\rho$, performing $\rhos$ (introduced above) on strings.
Clearly we have no guarantee that $\eta^{\cP}$ satisfies Eq.~\ref{eq:comFW}, hence, we have to (further) abstract the $\CFG$ in order to guarantee completeness w.r.t. the performed static analysis, namely in order to make it possible to perform the given static analysis on the code in the generated $\CFG$.
As observed in the previous section, in order to force completeness, we have to combine the desired abstraction $\ok{\eta^{\cP}}$ on predicates, with the code abstraction $\ok{\ov{\CAbs}[\rho]}$. 
Formally, in order to allow this operation, since $\ok{\eta^{\cP}}$ is defined on $\ok{\Psi^{\cP}}$, we have to restrict also $\ok{\ov{\CAbs}[\rho]}$ on $\ok{\Psi^\cP}$ (denoted $\ok{\ov{\CAbs}[\rho]^\cP}$). This abstraction is obtained by intersecting the meaning of each one of its elements (i.e., its concretization) with the set of predicates in the $\CFG$.
In the running example, we have to compute $\ok{\ov{\CAbs}[\mathsf{\rhoimp}]^{\cP_{\mathtt{str}}}}$, which is the code abstraction induced by the $\Sign$ on the predicates in $\ok{\cP_{\mathtt{str}}}$. For instance, all the predicates in  $\sset{\code{x:=s}}{s\in (5)^+}$ and the predicate \code{x:=1} cannot be distinguished when integers are abstracted by observing only their signs, hence the resulting abstraction is depicted in Fig.~\ref{fig:omega-eval}, where the abstract predicate \code{x:=}$\mathbb{Z}^+$ corresponds, in the concrete, to the set of predicates $\sset{\code{x:=s}}{s\in (5)^+}\cup\{\code{x:=1}\}$, while \code{x<}$\mathbb{Z}^+$ and $\neg$\code{(x<}$\mathbb{Z}^+$\code{)} correspond, respectively, to $\{\code{(x<5)}\}$ and to $\{\neg\code{(x<5)}\}$ (all the other elements corresponds to $\bot$).

Finally, we aim at building a code abstraction which can be interpreted by the initial abstract interpreter $\rho$, namely, that satisfies  Eq.~\ref{eq:comFW}. By Th.~\ref{th:omega} such an abstraction is $\ok{\ov{\eta}_\uparrow^\cP=\ov{\CAbs}[\rho]^\cP}\circ \eta^{\cP}$. 
\begin{corollary}
Let $\rho\in\uco(\cval)$ be additive. Then the code abstraction 
$\ok{\ov{\eta}_\uparrow^\cP=\ov{\CAbs}[\rho]^\cP\circ\eta^{\cP}\in\uco(\Psi^\cP)}$ is complete w.r.t.\ the semantic abstraction $\rho$, i.e., it satisfies Eq.~\ref{eq:comFW}.
\end{corollary}
\begin{figure}[t]
	\centering
	\includegraphics[scale=0.13]{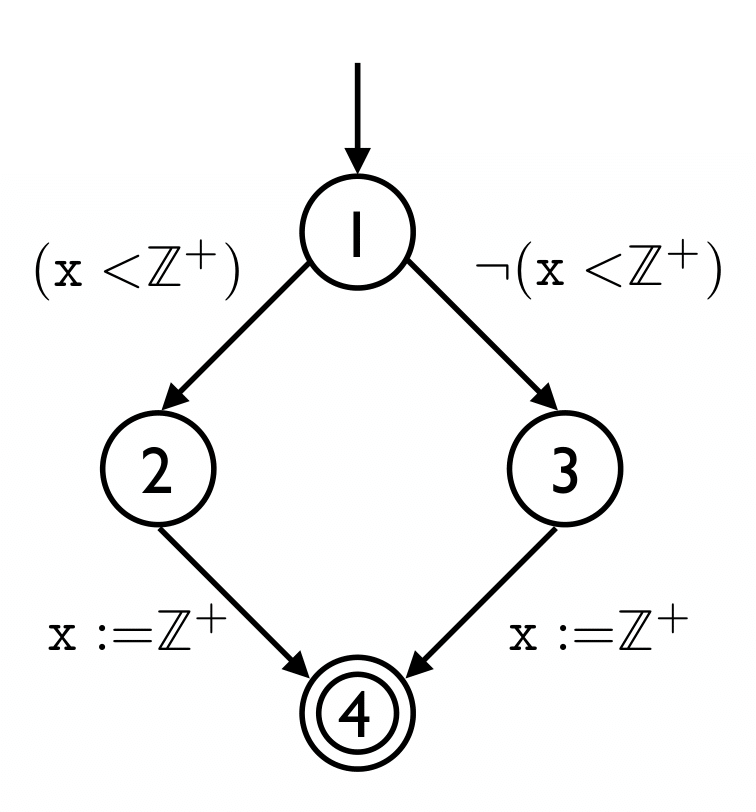}
	\caption{Abstract $\CFG$ generated by abstracting $\cP_{\mathtt{str}}$ by means of $\ov{\eta}_\uparrow^{\cP_{\mathtt{str}}}$}
	\label{fig:p-abs-abs}
\end{figure}

\noindent
Hence, in our example, the code abstraction $\ok{\ov{\eta}_\uparrow^{\cP_{\mathtt{str}}} =  \ov{\CAbs}[\mathsf{\rhoimp}]^{\cP_{\mathtt{str}}}}\circ\eta^{\cP_{\mathtt{str}}}$ satisfies Eq.~\ref{eq:comFW}. 
In particular, we can observe that $\ok{\ov{\eta}_\uparrow^{\cP_{\mathtt{str}}} = \ov{\CAbs}[\mathsf{\rhoimp}]^{\cP_{\mathtt{str}}}}$.
Finally, we have to abstract the $\CFG$ $\cP$, previously generated, by applying 
$\ok{\ov{\eta}_\uparrow^{\cP}}$ to each edge of the $\CFG$. In our example, the so far resulting abstract $\CFG$ is reported in Fig.~\ref{fig:p-abs-abs}, where the abstract $\CFG$ generated by abstracting $\cP_{\mathtt{str}}$ by means of $\ov{\eta}_\uparrow^{\cP_{\mathtt{str}}}$ is depicted.  
%
\paragraph*{A Taste of Implementation.}
A static analyzer based on finite state automata is available at~\cite{arceri2018}. Moreover, we have implemented Alg.~\ref{alg:autsfa} in order to validate our approach\footnote{Available at\\ \code{https://github.com/SPY-Lab/java-fsm-library/tree/abstract-parser}}. The implementation of a static analysis of abstract $\CFG$s is in an early stage development and it is left as future work. Nevertheless, it is able to parse executable automata and to abstract them into abstract $\CFG$s, as we have previously described. In order to make these abstract $\CFG$s effectively analyzable, we are currently extending the static analyzer, and the underlying abstract interpreter,  to parse, and thus analyze, also abstract predicates.

\section{Conclusion}\label{sec:conl}
We conclude by highlighting the value, in the context of static analysis, of the framework presented in this paper. 
What we propose here is a precision improvement of \cite{tops20}, an analysis that attacks an extremely hard problem in static program analysis by abstract interpretation, since the standard static analysis assumption (i.e., the program code we want to analyze must be static) is broken when we have to deal with string-to-code statements. In \cite{tops20}, we have shown that even without this assumption, it is still possible for static analysis to semantically analyze dynamically mutating code in a meaningful and sound way. It has been the very first proof of concept for a sound static analysis for self-modifying code based on bounded reflection for a high-level script-like programming language. 
In this paper, we improve this approach by characterizing code transformations that do not lose precision w.r.t.\ a fixed abstract semantics/analysis of the code. The idea we develop consists of embedding the property to analyze in the code transformation in order to make the property analysis work also on the transformed code (as it happens in dynamic code analysis). Hence, the main contribution is to make even more precise the first truly {\em dynamic static analyzer\/}, which has the feature to keep the analysis going on, even when code is dynamically built. \\
Clearly, the framework improved here is still at an early stage and surely there is much work to do, not only for the presented algorithm and the implementation, which has clearly to be further developed but also for making the approach more precise and general. 
As far as the algorithm is concerned we have not explicitly provided soundness and completeness proofs or discussions. In particular, completeness holds under decidable hypotheses (the input automaton has to recognize only executable strings), here only briefly treated, and therefore these aspects need further formal development.\\
On the other hand, a direction for improving precision can be that of integrating the proposed static analysis in a hybrid solution, by using, for instance, taint analysis (or other dynamic analyses) for driving when to apply static analysis, or considering more advanced forms of automata-based domains for abstracting strings, such as the one reported in~\cite{negrini21}. Finally, we have considered only \code{eval} as a string-to-code statement, while there are other ways, for dynamically executing code built out of strings, that should be investigated. However, we strongly believe that the same approach used for \code{eval}, could be easily applied to any other string-to-code statement.
Moreover, we believe that this framework could be instantiated in order to deal with other forms of code transformations, maybe by considering more general $\CFG$ abstractions.

From a more theoretical point of view, interesting future works consist of exploiting the proposed approach for analyzing code in order to investigate, on dynamic languages, several application contexts where static analysis by abstract interpretations has been exploited. First of all, we could trace (abstract) flows of information during execution \cite{GiacobazziM18,MastroeniZ17,Mastroeni13,MastroeniN10,GiacobazziM10,GiacobazziM10bis,GM04CSL} in order to tackle different security issues, such as the detection of (abstract) code injections \cite{BuroM18,MB10} or the formal characterization of dynamic code obfuscators and of their potency \cite{JGM12,GiacobazziM12}. Moreover, the ability to analyze malware code could be exploited for extracting code properties which could be used for analyzing code similarity \cite{PredaGLM15bis}, a technique useful for instance to identify or at least classify malicious code.

\bibliographystyle{eptcs}
\bibliography{biblio}


\end{document}